\documentclass[a4paper]{article}
\usepackage{fullpage}
\usepackage{dsfont,amsmath,amsthm,amssymb}

\usepackage{subfigure}

\usepackage{tikz}
\usetikzlibrary{shapes.symbols,arrows}
\usetikzlibrary{matrix}

\usepackage{enumerate}
\usepackage{fullpage}

\newtheoremstyle{newthm}
  {\topsep}   
  {\topsep}   
  {\itshape}  
  {0pt}       
  {\scshape} 
  {.}         
  {5pt}  
  {}          

\theoremstyle{newthm}

\newtheorem{thm}{Theorem}
\newtheorem{prop}[thm]{Proposition}
\newtheorem{lem}[thm]{Lemma}
\newtheorem{cor}[thm]{Corollary}

\newcommand{\IR}{\mathds{R}}
\newcommand{\IN}{\mathds{N}}

\newcommand{\IC}{\mathds{C}}

\newcommand{\onevec}{{\mathbf{1}}}

\renewcommand{\leq}{\leqslant}
\renewcommand{\ge}{\geqslant}
\renewcommand{\geq}{\geqslant}
\newcommand{\one}{\mathds{1}}

\usepackage{dsfont,mathtools}

\DeclareMathOperator\In{In}
\DeclareMathOperator\Out{Out}

\title{Approximate Consensus in Highly Dynamic Networks: \\
	 The Role of Averaging Algorithms}
\author{Bernadette Charron-Bost\textsuperscript{1} \and Matthias
F\"ugger\textsuperscript{2} \and Thomas Nowak\textsuperscript{3}}
\date{\textsuperscript{1} CNRS, \'Ecole polytechnique\\
\textsuperscript{2} TU Wien\\
\textsuperscript{3} ENS Paris}

\begin{document}
\maketitle

\begin{abstract}
In this paper, we investigate the approximate consensus problem in highly dynamic networks 
        in which topology may change continually and unpredictably.
We prove that  in both synchronous and partially synchronous systems,
        approximate consensus is solvable if and only if the
        communication graph in each round has a rooted spanning tree, i.e., there is a coordinator at each time.
The striking point in this result is that the coordinator is not required  to be unique and can
        change  arbitrarily from round to round.
Interestingly, the  class of averaging algorithms, which are memoryless
	and require no process identifiers,  entirely captures the solvability issue
        of approximate consensus in that the problem is solvable if and only if
        it can be solved using any averaging algorithm.

Concerning the time complexity of averaging algorithms, we show that approximate consensus can be 
	achieved with precision of $\varepsilon$ 
        in a coordinated network model in $O(n^{n+1} \log\frac{1}{\varepsilon})$ synchronous 
        rounds, and in $O(\Delta n^{n\Delta+1} \log\frac{1}{\varepsilon})$ rounds when the maximum 
        round delay for a message to be delivered is $\Delta$.
While in general, an upper bound on the time complexity of averaging algorithms has to be exponential,
        we investigate various network models in which this exponential bound in the number of nodes
        reduces to a polynomial  bound.

We  apply our results to networked systems with a fixed topology and classical benign  fault models,
        and deduce  both known and new results for approximate consensus in these systems.
In particular,  we show that for solving approximate consensus, a complete network can 
        tolerate up to $2n-3$ 
        arbitrarily located  link faults at every
        round, in contrast with the impossibility result 
         established by Santoro and Widmayer (STACS '89) showing that 
         exact consensus is not solvable with $n-1$ link faults per round originating from the same node.
\end{abstract}

\setcounter{footnote}{3}

\section{Introduction}

Recent years have seen considerable interest in the design of distributed algorithms
	for dynamic networked systems.
Motivated by the emerging applications of the Internet and mobile sensor systems,
	the design of distributed algorithms for networks with a swarm of nodes and  time-varying
	connectivity has been the subject of much recent work.
The algorithms implemented in such dynamic networks ought to be  decentralized, using local information,
	and resilient to mobility and link failures.

A large number of distributed applications  require to reach some kind of agreement in the network
	in finite time.
For example,  processes may attempt to agree on whether to commit or abort the results of
	a distributed database transaction; or sensors may try to agree on estimates of
	a certain variable; or vehicles may attempt to align their direction of motions with their 
	neighbors.
Another example is clock synchronization where processes attempt to maintain a common time scale.
In the first example, an {\em  exact consensus} is achieved on one of the outcomes (namely,  commit or
	abort) as opposed to the other examples where processes are required to agree on values
	that are sufficiently close to each other, but not necessarily equal.
The latter type of agreement is referred to as {\em approximate consensus}.
	
For the exact consensus problem, one immediately faces  impossibility results in truly dynamic networks
	in which some stabilization of the network during a sufficiently long period of time  is not assumed
	(see e.g.~\cite{SW89} and~\cite[Chapter 5]{Lyn96}).
Because of its wide applicability, the approximate consensus problem appears as an interesting weakening of
	exact consensus to circumvent these impossibility results.
The objective of the paper is exactly to study computability and complexity of approximate consensus 
	in dynamic networks in which the topology may change continually and in an unpredictable way.

\subsection{Dynamic networks}

We consider a fixed set of processes that operate in rounds and communicate 
	by broadcast.
In the first part of this article, rounds are supposed to be synchronous in the sense that
	the messages  received at some round  have been sent at that round.
Then we extend our results to partially synchronous rounds with a maximum allowable
	delay bound.

At each round, the communication graph is chosen arbitrarily among a  set
	of directed graphs that determines the network model.
Hence the communication graph can change continually and unpredictably from one
	round to the next.
Then the local algorithm at each process applies a state-transition function to its current state 
	and the messages received from its incoming neighbors
	in the current communication graph to obtain a new state. 

While local algorithms can be arbitrary algorithms in principle, the basic idea is to keep them simple
	so that coordination and agreement do not result from the local computational powers
	but from the flow of information across the network.
In particular, we focus on {\em averaging algorithms} which  repeatedly form convex combinations.
One main feature of averaging algorithms is to be memoryless in the sense that 
	the next value of each process is entirely determined only from the values
	of its incoming neighbors in the current communication graph.
More importantly, they work in anonymous networks, not requiring processes to have identifiers.

The  network model we consider unifies a wide variety of dynamic networks. 
Perhaps the most evident class of networks captured by this model is dynamic multi-agent networks
	in which communication links frequently go down while other links are established
	due to the mobility of the agents.
The dynamic network model can also serve as an abstraction for static or dynamic wireless
	networks in which collisions and interferences make it difficult to predict which messages 
	will be delivered in time.
Finally the dynamic network model  can be used to model traditional communication networks
	with a fixed communication graph (e.g., the complete graph) and some transient link failures.
	
In this model, the number of processes~$n$ is fixed, and we assume that each process
	knows~$n$.
This assumption can be weakened to a large extent.
Indeed all our results still hold when~$n$ is not the exact number of processes
	but only an upper bound on this number.
That allows us to extend the results to a completely dynamic network
	with  a maximal number of processes  that may join or
	leave.
This limit even disappears for the {\em asymptotic consensus} problem
	in which  processes are just required to converge to the same value.
	
Finally, for simplicity, we  assume that all processes start the computation at the same 
	round. 
In fact, it is sufficient to assume that every process eventually participates to the
	computation either spontaneously (in other words,  it initiates the computation)
	or by receiving, possibly indirectly, a message from an initiator.

\subsection{Contribution}

We make the following contributions in this work:

(i) The main result in this paper is the exact characterization of the network models
	in which approximate consensus is solvable.
We prove that the approximate consensus problem is solvable in a network model
	if and only if each communication graph in this model has a rooted spanning tree.
This condition guarantees that 
	the network has at least one coordinator in each round.
The striking point is that 
	coordinators may continually change over time without preventing nodes
	from converging to consensus.
Accordingly the network models in which 
	approximate consensus is solvable are called {\em coordinated network models}.
The proof of this computability result highlights the key role played by averaging
	algorithms in approximate consensus: the problem is solvable if and only if
	it can be solved using any averaging algorithm.

(ii) With  averaging algorithms, we show that agreement with precision of~$\varepsilon$ can be reached in 
	$O \left ( n ^{n  + 1} \log \frac{1}{\varepsilon} \right )$ rounds  
	in a coordinated network model, and in only $n \log \frac{1}{\varepsilon} $ rounds
	in the case of  a {\em nonsplit network model}, defined as a collection
	of communication graphs in which any two processes have at least one common incoming neighbor.
As a matter of fact, every general upper bound
	for the class of averaging algorithms has to be exponential since  the {\em equal neighbor} averaging algorithm 
	requires $\Omega(2^{n/3}\log\frac{1}{\varepsilon})$ 
	rounds to reach agreement with precision of~$\varepsilon$
	for the network model in~\cite{OT11v1}.

(iii) As an application, we revisit the problem of approximate consensus in the
	context of communication faults, whether they are due to  link or process failures.
We  first prove a new result  on the solvability of approximate consensus in a complete  
	network model in the presence of benign communication faults which shows that the number 
	of link faults that can be tolerated increases by a factor of at least 2 
	when solving  approximate consensus instead of consensus.
Then we prove the correctness of  fault-tolerant 
	approximate consensus  algorithms in a complete network by 
	interpreting them as averaging algorithms.
That allows us to extend the scope of these algorithms originally designed for the static crash failure model
	to a completely dynamic failure model. 
	
(iv) Finally we extend our computability and complexity results to the case of {\em partially synchronous rounds}
	in which communication delays may be non null, but are bounded by some positive
	integer $\Delta$: the messages  received at some round~$k$  have not necessarily been sent at round~$k$,
	but  at some round~$\ell$ with $\ell \in \{k-\Delta-1, \dots , k \}$.
We prove  the same necessary and sufficient condition on network models  for 
	solvability of approximate consensus, and  give an $O \left ( n ^{n\Delta  + 1} \log \frac{1}{\varepsilon} \right )$ 
	upper bound on the number of rounds needed by averaging algorithms to achieve agreement with precision of~$\varepsilon$ 
	in a coordinated network model.
For nonsplit network models, this bound reduces to the polynomial bound $O \left ( \Delta n^{2\Delta  -1 }  \log \frac{1}{\varepsilon}  \right )$.

\subsection{Related work}

Agreement problems have been extensively studied in the framework of static
	communication graphs or with limited topology 
	changes~(e.g.,~\cite{Lyn96,AW05, PBE07,Vai14}).
In particular, the approximate consensus problem, also called {\em approximate
	agreement}, is studied in numerous papers in the context of a complete 
	graph and at most $f$ faulty processes~(see, e.g.,~\cite{DLPSW86, Fek90,ALS94}).
In the case of benign failures where processes may crash or  omit to send some messages,
	the failure model yields communication graphs with a {\em fixed} core of at least $n-f$ processes
	that have outgoing links to all processes, and so play the role of
	steady coordinators of  the network.

There is also a large body of previous work  on general dynamic networks.
However, in much of them, topology changes are restricted and the sequences of
	communication graphs are supposed to be ``well-formed'' in  various senses.
Such well-formedness properties are actually opposite to the idea of
	unpredictable changes.
In~\cite{AFJ06}, Angluin, Fischer, and Jiang study the {\em stabilizing consensus problem}
	in which nodes are required to agree exactly on some initial value, but without necessarily 
	knowing  when agreement is reached, and they
	assume  that any two nodes can directly communicate infinitely often.
In other words, they suppose the limit graph formed by the links that occur infinitely often to be complete.
To solve the consensus problem, Biely, Robinson, and Schmid~\cite{BRS12} assume that 
	throughout every block of $4n-4$ consecutive communication graphs there exists 
	a stable set of roots.
Coulouma and Goddard~\cite{CG13} weaken the latter stability condition to obtain a characterization
	of the sequences of communication  graphs for which consensus is solvable.
Kuhn, Lynch, and Oshman~\cite{KLO10} study  variations of the {\em counting}  problem; they  
	assume bidirectional links and a stability property, namely the $T$-{\em interval
	connectivity} which stipulates that there exists a stable spanning tree over every $T$ consecutive 
	communication graphs.
All their computability results actually hold in the case of 1-interval connectivity which reduces to
	a property on the {\em set} of possible communication graphs, and the cases $T>1$ are 
	investigated just to improve complexity results.
Thus they  fully model unpredictable topology changes, at least for computability results
	on counting in a dynamic network.
	
The network model in~\cite{KLO10}, however, assumes a static set of nodes and communication graphs 
	that are all bidirectional  and connected.
The same  assumptions are made to study the time complexity of several variants of consensus~\cite{KMO11}
	in dynamic networks.
Concerning the computability issue,  such strong assumptions make exact agreement trivially solvable: 
	since communication graphs are continually strongly connected, nodes can collect the set of initial values
	and then make a decision on the value of some predefined function of this set
	(e.g., majority, minimum, or maximum).	

The most closely related pieces of work are doubtless those about asymptotic consensus
	and more specifically  {\em consensus sets} studied by Blondel and Olshevsky~\cite{BO13}:
	a consensus set is a set of stochastic matrices such that every infinite product of
	matrices from this set converges to a rank one matrix.
Computations of averaging algorithms correspond to infinite products of stochastic 
	matrices, and the property of asymptotic consensus is captured by  convergence to a rank one
	matrix.
Hence when an upper bound on the number of nodes is known, the general 
	notion  of  network models in which approximate consensus
	is solvable  reduces to the notion of consensus sets 
	if  we restrict ourselves to averaging algorithms.
However the characterization of consensus sets in~\cite{BO13}  is not  included into
	our main computability result for approximate consensus, namely Corollary~\ref{cor:solvac},
	since the fundamental assumption of a self loop at each node in  communication graphs
	(a process can obviously communicate with itself)
	does not necessarily hold for the directed graphs associated to  stochastic matrices in a consensus set.
The characterization of compact consensus sets in~\cite{BO13} and our computability result of
	approximate consensus  are thus incomparable.

In the same line, some of  our positive results which make use of averaging algorithms
	can be shown equivalent to results about stochastic matrix products in the
	vast existing literature on asymptotic consensus~\cite{DeG74,CS77,Mor05,AB06,
	CMA08a,LL10,TN11,LMMAY11,Cha13,XC12,HT13,Now13}.
Notably Theorem~\ref{thm:coord} 
	is  similar to the central result in~\cite{CMA08a}, but
	we give a different proof much simpler and direct as it requires neither
	star graphs (called {\em strongly rooted graphs}  in~\cite{CMA08a}) nor 
	Sarymsakov graphs~\cite{XC12}.
Moreover our proof yields a significantly better upper bound on the time complexity of averaging algorithms 
	in coordinated network models, namely $O \big (n^{n+1} \log \frac{1}{\varepsilon} \big)$	 instead
	of $O \big(n^{n^2} \log \frac{1}{\varepsilon} \big)$ in~\cite{CMA08b}.
The  statement in Theorem~\ref{thm:coord:delay} for the partial synchronous case with bounded delays
	already appears in~\cite{CMA08b,Cha13},
	but our proof strategy, which consists in a reduction to the case of synchronous nonsplit
	networks, yields  a new and simpler proof.

\section{Approximate consensus and averaging algorithms}

We assume a distributed, round-based computational model in the spirit
     of the Heard-Of model by Charron-Bost and Schiper~\cite{CS09}.
A system consists in a set  of processes $[n] = \{1,\dots,n\}$.
Computation proceeds in {\em rounds}:
In a round, each process sends its state to its outgoing neighbors,
	 receives values from its incoming neighbors, and
     finally updates its state based.
The  value of the updated state is determined by a deterministic
     algorithm, i.e., a transition function that maps the values in the
     incoming messages  to a  new state value.
Rounds are communication closed in the sense that no process receives
     values in round~$k$ that are sent in a round different from~$k$.

Communications that occur in a round are modeled by a directed graph~$G=([n], E(G))$ 
	with  a self-loop at each node.
The latter requirement is quite natural as a process can obviously communicate with 
	itself instantaneously.
Such a directed graph is called a  {\em communication graph}.
We denote by~$\In_p(G)$ the set
     of incoming neighbors of~$p$ and by $\Out_p(G)$ the set of
     outgoing neighbors of~$p$ in~$G$.
Similarly $\In_S(G)$ and  $\Out_S(G)$ denote the sets of the incoming and
	outgoing neighbors of the nodes in a non-empty set $S \subseteq [n]$.
Since there is a self-loop at each process,   
	$	S \subseteq \In_S(G) \cap \Out_S(G) $,
	and so  both $\In_S(G)$  and $\Out_S(G)$ 
	are  non-empty.
The  cardinality of~$\In_p(G)$, i.e.,  the number of incoming neighbors of~$p$,
	is called the {\em in-degree} of process $p$ in~$G$.

A {\em communication pattern\/} is a sequence~$(G(k))_{k\ge 1}$ of
     communication graphs.
For a given communication pattern, $E(k)$,
	$\In_p(k)$ and  $\Out_p(k)$ stand for
	$E \left( G(k) \right)$, $\In_p(G(k))$ and $\Out_p(G(k))$, respectively.

Each process~$p$ has a {\em local state\/} $s_p$ the values of which at
     the end of round~$k \ge 1$ is denoted by~$s_p(k)$.
Process~$p$'s initial state, i.e., its state at the beginning of round~$1$,
     is denoted by~$s_p(0)$.
Let the {\em global state} at the end of round~$k$ be the collection
     $s(k) = (s_p(k))_{p \in [n]}$.
The {\em execution\/} of an algorithm from global initial state~$s(0)$,
     with communication pattern~$(G(k))_{k\ge 1}$ is the unique
      sequence $(s(k))_{k\ge 0}$ of global states defined as follows:  for  each round~$k \ge 1$, 
      process~$p$ sends~$s_p(k-1)$ to all the processes in~$\Out_p(k)$,
     receives $s_q(k-1)$ from each process~$q$ in~$\In_p(k)$, and
     computes~$s_p(k)$ from the incoming messages, according to the
     algorithm's transition function.
 
\subsection{Consensus and approximate consensus}
 
A crucial problem in distributed systems is to achieve agreement among local process states
     from arbitrary initial local states.
It is a well-known fact that this goal is not easily achievable in the context 
	of dynamic network changes~\cite{FLP85,SW89}, and restrictions on communication patterns
	are required for that.
We thus define a {\em network model\/} as a non-empty set ${\cal N}$ of
     communication graphs, those that may occur in communication patterns.

We now consider the above round-based algorithms in which the local state of 
	process~$p$ contains two variables~$x_p$ and~$dec_p$. 
Initially the range of $x_p$ is $[0,1]$
 	and $dec_p = \bot$ (which informally means that $p$
	has not  decided).\footnote{%
	In the case of  {\em binary consensus\/}, $x_p$  is restricted to be initially from $\{0,1\}$.}
Process~$p$ is allowed to set~$dec_p$ to the current value of $x_p$, and so to a value~$v$ different from $\bot$,
	     only once; in that case we say that~$p$ {\em decides}~$v$.	
An algorithm {\em  achieves consensus with the communication pattern}~$(G(k))_{k\ge 1}$ 
	 if each execution from a global initial state as specified above and
             with the communication pattern~$(G(k))_{k\ge 1}$   fulfills the following three conditions:               
	\begin{description}
	  \item{\em Agreement.} The decision values of any two processes are equal.

	  \item{\em Integrity.} The decision value of any process is an initial value.

	  \item{\em Termination.} All processes eventually decide.

	\end{description}
An algorithm {\em  solves consensus\/} in a network
	     model~${\cal N}$ if it achieves consensus with each communication
	     pattern formed with graphs all in ${\cal N}$.
Consensus is {\em solvable in a network model~${\cal N}$\/} if there
     exists an algorithm that solves  consensus in~${\cal N}$.
Observe that consensus is solvable in~${\cal N}$ in $n-1$ rounds if each communication graph 
	in~${\cal N}$ is strongly connected. 
The following impossibility result  due to Santoro and Widmayer~\cite{SW89}, however, shows that network
     models in which consensus is solvable are highly constrained:
     consensus is not solvable in some ``almost complete'' graphs.
     
\begin{thm}[\cite{SW89}]\label{thm:SW}
Consensus is not solvable in the network model comprising all
     communication graphs in which at least~$n-1$ processes have
     outgoing links to all other processes.
\end{thm}

The above theorem is  originally stated by Santoro and Widmayer  in the context of link faults
	in a complete communication graph but its scope can be trivially extended to dynamic 
	communication networks.

To circumvent the impossibility of consensus even in such
     highly restricted network models, one may  weaken Agreement into 
 \begin{description}
\item{\em $\varepsilon$-Agreement.} The decision values of any two processes
  are within an {\it a priori} specified~$\varepsilon > 0$.
\end{description}
and replace  Integrity by:       
\begin{description}
\item{\em Validity.}  All decided values are in the range of the initial values of
     processes.
\end{description}

An algorithm {\em  achieves $\varepsilon$-consensus with the communication pattern}~$(G(k))_{k\ge 1}$ 
	 if each execution from a global initial state as specified above and
             with the communication pattern~$(G(k))_{k\ge 1}$   fulfills Termination, 
             Validity, and $\varepsilon$-Agreement. 
 An algorithm {\em  solves approximate consensus\/} in a network
	     model~${\cal N}$ if for any $\varepsilon >0$, it achieves $\varepsilon$-consensus with each communication
	     pattern formed with graphs all in~${\cal N}$.
Approximate consensus is {\em solvable in a network model~${\cal N}$\/} if 
	there exists an algorithm that solves   approximate consensus in~${\cal N}$.

\subsection{Averaging algorithms}

We now focus on  {\em averaging algorithms} defined by the update rules 
	for each variable $x_p$ which are of the form:
	\begin{align}
 	x_p(k) = \sum_{q \in \In_p(k)} w_{qp}(k) \, x_q(k-1) ,\label{eq:update}
	\end{align}
where $w_{qp}(k)$ are  positive real numbers with  $\sum_{q \in \In_p(k)} w_{qp}(k) = 1$.
 In other words,  at each round~$k$, process~$p$ updates~$x_p$ to some weighted average
	of the values~$x_q(k-1)$ it has just received.
For convenience, we let $w_{qp}(k) = 0$ if $q$ is not an incoming neighbor of $p$ in the
	communication graph of round~$k$.

An {\em averaging algorithm with parameter} $\varrho >0$ is an averaging algorithm with
	the positive weights uniformly lower bounded by~$\varrho\,$: 
	$$ \forall k\ge 1,\, p,q \in [n] \,:\, w_{qp}(k) \in \{0\}\cup[ \varrho , 1] \,.$$
Since we strive for distributed implementations of averaging
     algorithms, $w_{qp}(k)$ is required to be locally computable.
Finally note that the decision rule is not specified in the above definition of averaging algorithms:
	the decision time immediately follows from the number of rounds that is proven 
	to be sufficient to reach $\varepsilon$-Agreement.

Some averaging algorithms with locally computable weights are of particular interest, namely,
	the {\em equal neighbor averaging algorithm\/} and the {\em fixed weight averaging algorithms}.
      
In the equal neighbor averaging algorithm, at each round~$k$ process~$p$ chooses 
	\begin{equation}\label{eq:EN}
	w_{qp}(k) = 1/|\In_p(k)|
	\end{equation}
	for every $q$ in $\In_p(k)$.
It is clearly an averaging algorithm  with parameter $\varrho = 1/n$.

Given a network model ${\cal N}$, we denote by~$d_p^-\left ({\cal N} \right )$ the  maximum
	in-degree of process~$p$ over all the graphs in ${\cal N}$.
Each process~$p$ is  {\em a priori\/} assigned a positive parameter~$\alpha_p \geq d^-_p\left ({\cal N} \right )$.
In a  fixed weight averaging algorithm, at  every round~$k$, 
	process~$p$ chooses  
\begin{align}
w_{qp}(k) =  \begin{cases}
            1/\alpha_q & \text{ if } q \neq p\,,\\
            1 - \sum_{q\in \In_p(k) \setminus \{p\}} 1/\alpha_q & \text{ if } q=p\,.
            \end{cases}
\end{align}
	for each $q$ in $\In_p(k)$.
We verify that 
	$\varrho =  \min\{1/\alpha_p \mid p \in [n] \} $ is a positive lower bound on  positive weights.

\section{Solvability of approximate consensus}

In this section, we characterize the network models in which approximate consensus is 
	solvable.
For that we first prove that if any two processes have a common incoming neighbor 
	in each communication graph of a network model ${\cal N}$, then every averaging algorithm 
	solves approximate consensus in~${\cal N}$.
Then we extend this result to {\em coordinated network models} where each communication graph has a
	spanning tree with a root that plays the role of a {\em coordinator}. 
The latter result which is quite intuitive in the case of a fixed coordinator, actually holds when
	coordinators vary over time.
Conversely we show that  if approximate consensus is solvable in  ${\cal N}$,
	then ${\cal N}$ is necessarily a coordinated network  model.

\subsection{Nonsplit network model}

We say a directed graph~$G$ is {\em nonsplit\/} if for all pairs of
     processes~$(p,q) \in [n]^2$, it holds that
     \begin{equation}
       \In_p(G)  \cap \In_q(G)  \neq \emptyset   \,.
     \end{equation}
Accordingly we define a {\em nonsplit network model\/} as a network model 
	in which each communication graph is nonsplit.
Note that a special case of a nonsplit communication graph is one in
     which all processes have one common incoming neighbor, i.e.,
     hear of at least one common process~$r$.

Intuitively, the occurrence of a nonsplit communication graph makes
	 the variables~$x_p$   in an averaging algorithm to come closer  
	 together: by definition of nonsplit communication graphs, any two 
	 processes~$p$ and~$q$ have at least one common incoming neighbor~$r$, 
	 leading to a common term in both $p$'s and~$q$'s average.
The following theorem  formalizes this intuition, showing
     that approximate consensus is achieved in nonsplit network models.  

\begin{thm}\label{thm:nonsplit}
In a nonsplit network model of $n$ processes, every averaging algorithm with parameter~$\varrho$
	achieves $\varepsilon$-consensus in 
	$  \frac{1 }{\varrho} \log  \frac{1}{\varepsilon} $ rounds.
In particular, the equal neighbor averaging algorithm achieves $\varepsilon$-consensus in 
	$n \log \frac{1}{\varepsilon} $ rounds. 
\end{thm}

\begin{proof}
Validity is trivially satisfied by definition of an averaging algorithm.

For $\varepsilon$-Agreement, we first observe that the set of update rules~(\ref{eq:update}) can be concisely rewritten as 
	\begin{equation}\label{eq:recurrence}
	x(k) = W(k) x(k-1) 
	\end{equation}
	where $x(k)$ is the vector in $\IR^n$ whose $p$-th entry is the value held by process~$p$
	at time $k$ and $W(k)$ denotes the $n \times n$ matrix whose entry at the $p$-th row and $q$-th
	column is equal to 
	$$W_{p q}(k) = w_{q p}(k) .$$
By definition of an averaging algorithm, each matrix $W(k)$ is stochastic and
	 the positive entries are lower bounded by~$\varrho \in ]0,1]$. 

For any positive integers~$k$ and~$\ell$, $\ell  \geq k$, we let 
	$$W(\ell : k )= W(\ell)\dots W(k) .$$
In particular  $W(k : k )=  W(k) $. 
From the recurrence relation (\ref{eq:recurrence}) we derive 
	\begin{equation}\label{eq:expression}
	x(k) = W(k:1)x(0) \, .
	\end{equation}

Since each $G(k)$ is a nonsplit communication graph, we obtain that
     for any two processes~$p,q$ there is a process~$r$ with  
$$ \min (W_{p r}(k), W_{q r}(k)) \ge \varrho \,.$$

The coefficient of ergodicity of a stochastic matrix $P$ introduced by Dobrushin~\cite{Dob56}
	and defined by 
	$$\delta(P) = 1- \min_{p,q} \sum^n_{r=1} \min (P_{p r}, P_{q r}) $$
	thus satisfies the inequality 
	\begin{equation}\label{eq:normC}
	\delta \big( W(k) \big)  \leq 1 - \varrho .
	\end{equation}
Besides a result by Seneta~\cite{Sen79} combined with a straightforward argument of convex duality shows 
	that for any stochastic matrix~$P$, the coefficient $\delta(P)$ coincides with the matrix seminorm 
	 $$ \sup _{x\notin \IR \one}  \frac{ \delta (Px) } {\delta(x)} $$
	 associated to  the seminorm on $\IR^n$ defined by
	$\delta(x)= \max_{p} (x_p) - \min_{p} (x_p) $, where $ \IR \one$ is the line
	of vectors with equal components.
Consequently $\delta$ is a matrix seminorm, and so is sub-multiplicative.

Since  $\delta(P) \leq 1$ for any stochastic matrix $P$, we conclude that
	$$ \delta \left( W( k:1)\right) \leq \left (1 - \varrho \right)^{k} \, .$$
Because of the inequality $$ 1-a \leq e^{-a}$$  when 
	$a \geq 0$  and
	because $ \delta \left( x(0 )\right) \leq 1$, it follows that 
		if $k \geq \frac{1 }{\varrho} \log \frac{1}{\varepsilon}$,
		then $ \delta \left( x(k)\right) \leq\varepsilon$.
This completes the proof of Theorem~\ref{thm:nonsplit}.
\end{proof}

For technical purposes, we now extend Theorem~\ref{thm:nonsplit} in two directions.
We first observe that the 
	above proof does not make use of the assumption of a self-loop at each node of
	communication graphs.
Therefore Theorem~\ref{thm:nonsplit}  still holds for {\em generalized network models\/} 
	in which each process does not necessarily communicate with itself.
The second extension concerns the granularity at which the assumption of nonsplit 
	communication graphs holds.
Let the {\em product\/} of two directed graphs $G$ and $H$ with the same set of nodes $V$
	be the directed graph  $G \circ H$ with set of nodes~$V$ and a  link from $p$ to $q$ 
	 if there exists a node~$r$ such that $(p,r)\in E(G)$ and $ (r,q) \in E(H)$.
For any positive integer~$K$, we say a network model~${\cal N}$ is $K$-{\em nonsplit\/} if any product of $K$
	graphs from~${\cal N}$ is nonsplit. 

\begin{cor}\label{cor:nonsplit}
In a generalized $K$-nonsplit network model of $n$ processes, every averaging algorithm with parameter~$\varrho$
	achieves $\varepsilon$-consensus in 
	$ K \left (\frac{1 }{\varrho}\right )^{K} \log \frac{1}{\varepsilon} + \, K -1$ rounds.
\end{cor}

\begin{proof}
We repeat the beginning of the proof of Theorem~\ref{thm:nonsplit}
	and we form the matrix product $W(k:1)$.
When grouping matrices $K$ by $K$,  $W(k:1)$ turns out to be  a product
	of $\lfloor \frac{k}{K} \rfloor$ blocks of the form $W(\ell +K -1 :\ell)$
     and at most $K-1$ remaining stochastic matrices.
Each block matrix is nonsplit and 
	its positive entries are lower bounded by $\varrho^K$.
Hence for any two processes~$p,q$ there is a process~$r$ with  
	$$ \min (W_{p r}( \ell +K - 1: \ell ), W_{q r}( \ell +K - 1: \ell )) \ge \varrho^K \,.$$
It follows that the coefficient of ergodicity $\delta$ of each block  satisfies
	$$
	\delta \big( W( \ell +K - 1: \ell) \big)  \leq 1 - \varrho^K .
	$$
By the sub-multiplicativity of~$\delta$ and since  $\delta(P)\leq 1$ when 
	$P$ is a stochastic matrix, 
     we obtain that 
	$$ \delta \left( W( k:1)\right) \leq \left (1 - \varrho^{K} \right)^{\lfloor k/K\rfloor} \, .$$
Since $\delta \left ( x(0) \right )\leq 1$, it follows  that 
     if $k \geq K \left (\frac{1 }{\varrho}\right )^{K} \log \frac{1}{\varepsilon} + \, K-1$,
     then $ \delta \left( x(k)\right) \leq\varepsilon$.
We conclude as in the proof of Theorem~\ref{thm:nonsplit}.
\end{proof}

\subsection{Coordinated network model}\label{sec:coordinated}

We begin by recalling some basic notions on directed graphs.
A directed graph~$G$ is said to be $p\,$-{\em rooted}, for some node $p$,
 	if for every node there exists a directed path  terminating at this node and originating at~$p$. 
Such a node~$p$ is called a {\em root of} $G$, and $R_G$ denotes the  set of roots in~$G$.
If~$R_G$ is non-empty, then $G$ is said to be {\em rooted}.

\begin{prop}\label{prop:rooted}
Let   $S$ be a non-empty set of nodes of a directed graph~$G$.
If $S$ has no incoming link, then $S$ contains all the roots of $G$, i.e., 
	$$ \In_S(G) \subseteq  S \ \Rightarrow \  R_G \subseteq S  \, .$$
\end{prop}

The {\em condensation\/} of a directed graph~$G$, denoted by $G^*$, is the directed graph of 
	the strongly connected components of $G$.
Clearly a directed graph~$G$ is rooted if and only if its condensation $G^*$ is.
Using that~$G^*$  is acyclic, we  can show the following characterization of rooted graphs.
 
\begin{prop}\label{prop:rootedbis}
A directed graph~$G$ is rooted if and only if its condensation~$G^*$ has a sole node without
        incoming neighbors.
\end{prop}
\noindent As a consequence of the above proposition, a nonsplit directed graph is rooted.

%

\medskip
Intuitively, while communication graphs remain $p$-rooted, process~$p$ gathers the 
	values in its strongly connected component, computes  some weighted average value,
	and attempts to impose this value to the rest of the processes.
In other words, its particular position in the network makes~$p$ 
	to play the role of {\em network coordinator\/} in  any averaging algorithm.
Accordingly, we define a {\em coordinated network model\/} as a network model 
	in which each communication graph is rooted.

From the above discussion, it is easy to grasp why in the particular 
	case of a steady coordinator, all processes converge to a common value and so achieve approximate 
	consensus when running an averaging algorithm.
We now  show that approximate consensus is actually achieved even when coordinators
	change over time. 
For that, we begin with the following elementary lemma.

\begin{lem}\label{lem:productrooted}
For each system with $n$ processes, any coordinated network model  is $(n-1)$-nonsplit.
\end{lem}

\begin{proof}
Let ~$H_1, \dots, H_{n-1}$ be a sequence of $n-1$ communication graphs,
	each of which is rooted.
For each process $p \in [n]$ and each index $ k \in \{0, \dots, n-1\}$,  we define the sets~$S_p(k)$ by  
	\begin{equation}\label{eq:recSp}
         S_p(0) = \{ p \} \ \text{ and } \ 
         S_p(k) = \In_{S_p(k-1)}  (H_k )   \text{ for } k \in \{1, \dots, n-1\} \,.
         \end{equation}
We easily check that for any $k \in \{1, \dots, n-1\} $, 
	\begin{equation}\label{eq:Sp}
         S_p(k) = \In_p(H_k \circ \dots \circ H_1 ) \, .
         \end{equation}
Because of the self-loops at all nodes in communication graphs,
	 $S_p(k) \subseteq S_p(k+1) $.
Hence none of the sets $S_p(k)$ is empty. 

We now show that for any processes $p,q \in [n]$,
\begin{equation}\label{eq:intersect}
S_p(n-1) \cap S_q(n-1) \neq \emptyset \,.
\end{equation}
If $p=q$, then~(\ref{eq:intersect}) trivially holds.
Otherwise, assume by contradiction that \eqref{eq:intersect} does not hold;
	it follows that for any index $k  \in \{ 0, \dots, n-1\}$, the sets $S_p(k)$ and  $S_q(k)$ are disjoint.
	
Let us consider the  sequences $ S_p(0) \subseteq \dots  \subseteq S_p(n -1)$,
	$S_q(0) \subseteq  \dots  \subseteq S_q( n -1) $, and 
	$$ S_p(0) \cup S_q (0) \subseteq  \dots  \subseteq S_p( n -1) \cup S_q( n -1) \, .$$
Because  $ |S_p(0) \cup S_q(0)| \ge 2 $ if $p\neq q$ and
	  $ |S_p(n-1) \cup S_q( n-1)| \leq  n $, the latter sequence cannot be strictly
	increasing.
Therefore there exists some index $\ell   \in \{ 0, \dots, n-2\}$
	such that $$ S_p(\ell   ) \cup S_q( \ell  ) = S_p(\ell   +1) \cup S_q(\ell   +1)\, .$$
By assumption, we have  $S_p( \ell   ) \cap S_q( \ell   ) = \emptyset$ and $S_p(\ell   +1 ) \cap S_q( \ell   +1 ) = \emptyset$.
Hence $$ S_p(\ell   ) = S_p(\ell   +1) \mbox{ and }  S_q( \ell  )  = S_q(\ell   +1)\, .$$
By (\ref{eq:recSp}) and Proposition~\ref{prop:rooted}, both $S_p(\ell   ) $ and $ S_q( \ell  )$ contain the
       nonempty set of  roots of $H_{\ell   +1}$, a contradiction to the disjointness assumption. 
Thus~(\ref{eq:intersect}) follows. 

\noindent Because of (\ref{eq:Sp}), this proves that the directed graph $ H_{n-1} \circ \dots \circ H_1  $ is nonsplit.
\end{proof}

We now apply Corollary~\ref{cor:nonsplit} to obtain  the following result.

\begin{thm}\label{thm:coord}
In a coordinated network model of $n$ processes, every averaging algorithm with parameter~$\varrho$
	achieves $\varepsilon$-consensus in 
	$ \left (\frac{1 }{\varrho}\right )^{n} n \log \frac{1}{\varepsilon}  + \, n -1 $ rounds.
In particular, the equal neighbor averaging algorithm achieves $\varepsilon$-consensus in 
	$O \left ( n ^{n  + 1} \log \frac{1}{\varepsilon} \right )$ rounds. 
\end{thm}

\begin{cor}\label{cor:solvac}
The approximate consensus problem is solvable in any coordinated network model.
\end{cor}

Interestingly Lemma~\ref{lem:productrooted} corresponds to a {\em uniform translation} in the
	Heard-Of model~\cite{CS09} that transforms  each block of $n-1$ consecutive rounds with
	rooted communication graphs into one  macro-round with a  nonsplit  communication graph.
If each process applies an equal neighbor averaging procedure only at the end of each macro-round instead of
	applying it round by round, the resulting distributed algorithm, which is no more an averaging 
	algorithm, achieves $\varepsilon$-consensus in only
	$ O \left ( n ^{2} \log \frac{1}{\varepsilon} \right ) $ rounds.
	 
\subsection{Necessity for the coordinated model to solve approximate consensus}

As we now show, there exists an algorithm, whether or not it is an averaging algorithm, achieving
	approximate consensus in  a network model ${\cal N}$ only if  ${\cal N}$ is coordinated.

\begin{thm}\label{thm:noncoord}
In any non-coordinated network model, the approximate consensus problem is not solvable.  
\end{thm}

\begin{proof}
 We proceed by contradiction and we assume that there exists an algorithm that solves
	approximate consensus in a non-coordinated network model ${\cal N}$.
	
Let $G$ be a communication graph in  ${\cal N}$ that is not rooted.
Then by Proposition~\ref{prop:rootedbis}, the condensation~$G^*$  has at least two nodes without 
	incoming neighbors.
Let $P$ and $Q$ denote the set of processes  corresponding to two such nodes in $G^*$, i.e., 
	to two strongly connected components of $G$ without incoming links.
	
Then we consider three executions of the algorithm which share the same communication pattern,
	namely the sequence with the fixed graph $G$, and which differ only in the initial state:
	in the first execution all processes start with 0, in the second one they all start with 1, and in 
	the third every  process in $P$ starts with 0 while all the others --- including processes in $Q$ ---
	start with 1.
By Validity, all processes finally decide 0 and 1 in the first and the second execution, respectively.
Moreover the first and the third executions are indistinguishable from the viewpoint of each process in $P$;
	in particular each process in $P$  makes the same decision, namely 0, in both of these executions.
Similarly each process in $Q$  makes the same decision, namely 1, in the second and the third execution.
Therefore the third execution violates $\varepsilon$-agreement as soon as $\varepsilon<1$, a contradiction
	with the assumption that  the algorithm solves approximate consensus.
\end{proof}	

\begin{cor}\label{cor:characterization}
The approximate consensus problem is solvable in a network model ${\cal N}$ if and only if ${\cal N}$ is 
	a coordinated model.
\end{cor}

At the risk of oversimplifying, one might say that understanding averaging algorithms
	is understanding backward products of stochastic matrices,
	$$W(k:1) = W(k)W(k-1)\dots W(1) $$
	as $k$ grows to infinity.
The graph associated to a stochastic matrix $W$ of size $n \times  n$ is the directed graph $G(W)$
	with the set of vertices equal to $[n]$  and a set of directed edges~$E$ defined by 
	$$ (p,q) \in E \Leftrightarrow W_{p q} >0 \, .$$
The  communication graph at round~$k$ thus coincides with the graph associated to~$W(k)$. 
Following the terminology in~\cite{BO13}, a set ${\cal P}$ of stochastic matrices is a {\em consensus set} 
	if every infinite backward product of matrices from ${\cal P}$ converges to a rank one matrix.
When limiting ourselves to averaging algorithms, Corollary~\ref{cor:characterization} then reduces to a
	necessary and sufficient condition on a compact set of stochastic matrices with positive diagonals
	 to form a consensus set.
In particular, we obtain a graph-based characterization of stochastic matrices with positive diagonals 
	whose powers converge to a rank one matrix.

\begin{cor}\label{cor:characterizationmath}
A compact set ${\cal P}$ of stochastic matrices with positive diagonals is a consensus set if and only if 
	the directed graph associated to each matrix in ${\cal P}$ is rooted.
\end{cor}

\section{Time complexity of averaging algorithms}

As opposed to the approximate consensus algorithm sketched in Section~\ref{sec:coordinated}, 
	based on the translation of coordinated rounds into a nonsplit  macro-round, 
	one main advantage of averaging algorithms is that they do not require processes to have identifiers. 
Unfortunately the upper bound on the decision times of averaging algorithms in Theorem~\ref{thm:coord} is quite large, 
	namely exponential in the number of processes, while the decision times of the approximate consensus algorithm  in Section~\ref{sec:coordinated} 
	are at most quadratic.
Our goal in this section is precisely to study the time complexity of averaging  algorithms for
	a coordinated network model of  anonymous processes.

As shown in Theorem~\ref{thm:nonsplit}, the assumption of dynamic
     nonsplit networks drastically reduces the decision time  of
     averaging algorithms: for instance, the decision time of the
     equal neighbor algorithm is actually linear.
Another class of network models with polynomial decision times for
     some averaging algorithms are the bidirectional connected network
     models, i.e., those whose network model is a set of bidirectional
     and connected networks: from a result by Chazelle~\cite{Cha13},
     we derive a fixed weight averaging algorithm that achieves
     $\varepsilon$-consensus in $O(n^3 \log\frac{1}{\varepsilon})$
     rounds in these dynamic networks.
However in the non-bidirectional case, we show that the same algorithm
     may exhibit an $\Omega(2^{n/3}\log\frac{1}{\varepsilon})$
     decision time, demonstrating that a general upper bound on
     decision times in averaging algorithms has to be exponential.

\subsection{Bidirectional connected network models}

The consensus algorithm in~\cite{BRS12} with a linear decision time can be used in the context of bidirectional connected 
	networks since any process is then a root, and the set of roots is the set of all processes.
However  contrary to any  averaging algorithm, it does not tolerate
	deviations from bidirectional connected network models: for instance, it does not tolerate any link removal
	in the case of a bidirectional tree.
In other words, even if linear, the algorithm in~\cite{BRS12}  is not relevant for approximate consensus
	in the setting of dynamic networks that are bidirectional only most of the time.

Interestingly in any bidirectional connected network model with $n$
     processes, a fixed weight algorithm with parameter $\varrho$
     achieves $\varepsilon$-consensus in $O(\frac{1}{\varrho} n^2
     \log\frac{1}{\varepsilon})$ rounds for any $0 <\varepsilon
     <\varrho/n$ (see e.g., Theorem 1.6 in~\cite{Cha13}).
Hence in such network models, there exist fixed weight averaging
     algorithms that solve approximate consensus with polynomial
     decision times.
	
The proof in~\cite{Cha13} is based on some classical spectral gap arguments.
Let $W$ be any stochastic matrix; its  spectral radius is then equal to 1.
By the Perron-Frobenius theorem, the eigenvalue~1  is actually a simple 
	eigenvalue with positive eigenvectors if the matrix~$W$ is {\em primitive}, i.e., 
	the directed graph defined by
	its positive entries is strongly connected and aperiodic.
If $W$ is a primitive matrix, then  its transpose $W^T$ is also  primitive with the same spectral radius, 
	namely 1.
It follows that  1 is a simple eigenvalue of  $W^T$ with some  positive
	eigenvectors.
Hence there exists a unique positive vector $\pi$, called the {\em Perron vector of} $W$,
	such that $W^T \pi = \pi$  and  $\sum_p \pi_p = 1$.
Chazelle observes that  the Perron vectors of the stochastic matrices~$W(k)$
	associated to the execution of a fixed weight averaging algorithm with a communication
	pattern composed of bidirectional  graphs
	are constant, even though the network topology may change over time.
Using the inner product on $\IR^n$  defined by
	$$\langle x, y \rangle_{\pi} = \sum_{p=1}^n \pi_p \, x_p \, y_p $$
	where $\pi$ is the common Perron vector of the stochastic matrices~$W(k)$
	and the fact that each matrix~$W(k)$ is self-adjoint with respect to this inner product, 
	Chazelle establishes some bounds on the spectral gap of every matrix $W(k)$ when the communication 
	graph at round~$k$  is additionally connected,
	which allows him to conclude.
	
\subsection{Exponential decision time in unidirectional networks}

We develop the example given by Olshevsky and Tsitsiklis~\cite{OT11v1},
        inspired by~\cite{Chu97},
        of a strongly connected uni-directional network model and a fixed weight
	averaging algorithm that, in contrast to the bidirectional case, exhibits a
	necessarily exponentially large decision time.\footnote{We would like to thank Alex Olshevsky for
        pointing us to this example.}
More specifically, we show that the algorithm achieves $\varepsilon$-consensus in
	$\Omega(2^{n/3}\log\frac{1}{\varepsilon})$ rounds.
The example does not even require the network to be dynamic, using a time-constant network
        only.
The fixed weight averaging algorithm used in the example corresponds to the equal neighbor algorithm
        for the considered static network.

The communication graph, that we call the {\em $m$-butterfly graph}, is depicted 
	in Figure~\ref{fig:ex:exp}.
It has $n = 2m$ processes and consists of two isomorphic parts that are connected
	by a bidirectional link.
We  list the links between the processes $1,2,\dots,m$, which also determine
	the links between the processes $m+1,m+2,\dots,2m$ via the isomorphism 
	$\bar{p} = 2m - p + 1$.
The links between the processes $1,2,\dots,m$ are: (a)~the links $(p+1,p)$ for all
	$p \in [m-1] $ and (b)~the links $(1,p)$ for all $p\in[m]$.
In addition, it contains a self-loop at each process and the two 
	 links $(m,\bar{m})$ and $(\bar{m},m)$. 
Hence the $m$-butterfly graph is strongly connected.
	 
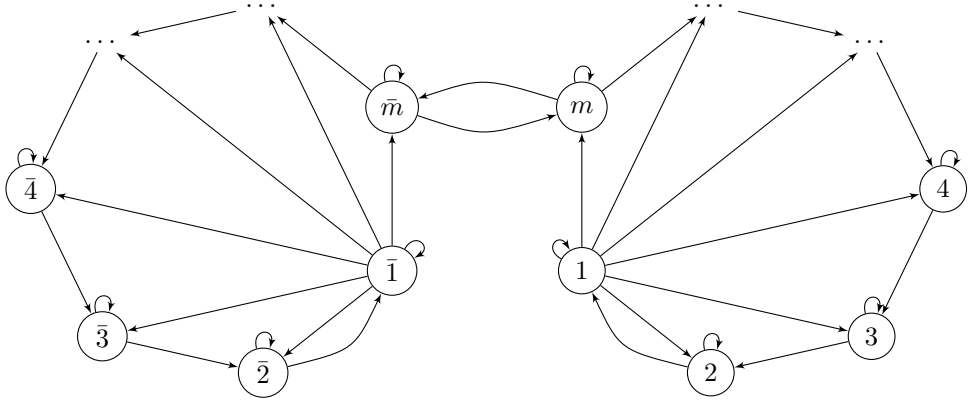
\begin{figure}
\centering
\begin{tikzpicture}[>=latex']
\node[draw, circle] (n1) at (205.71:2.5) {$1$};
\node[draw, circle] (n2) at (257.14:2.5) {$2$};
\node[draw, circle] (n3) at (308.57:2.5) {$3$};
\node[draw, circle] (n4) at      (0:2.5) {$4$};
\node               (n5) at  (51.43:2.5) {$\dots$};
\node               (n6) at (102.86:2.5) {$\dots$};
\node[draw, circle] (n7) at (154.29:2.5) {$m$};

\draw[->] (n1) -- (n2);
\draw[<-] (n2) -- (n3);
\draw[<-] (n3) -- (n4);
\draw[<-] (n4) -- (n5);
\draw[<-] (n5) -- (n6);
\draw[<-] (n6) -- (n7);
\draw[<-] (n7) -- (n1);

\draw[->] (n2) .. controls +(-1.2,0.3) .. (n1);
\draw[<-] (n3) -- (n1);
\draw[<-] (n4) -- (n1);
\draw[<-] (n5) -- (n1);
\draw[<-] (n6) -- (n1);

\begin{scope}[yscale=1,xscale=-1,shift={( 7,0)}]
\node[draw, circle] (m1) at (205.71:2.5) {$\bar1$};
\node[draw, circle] (m2) at (257.14:2.5) {$\bar2$};
\node[draw, circle] (m3) at (308.57:2.5) {$\bar3$};
\node[draw, circle] (m4) at      (0:2.5) {$\bar4$};
\node               (m5) at  (51.43:2.5) {$\dots$};
\node               (m6) at (102.86:2.5) {$\dots$};
\node[draw, circle] (m7) at (154.29:2.5) {$\bar m$};

\draw[->] (m1) -- (m2);
\draw[<-] (m2) -- (m3);
\draw[<-] (m3) -- (m4);
\draw[<-] (m4) -- (m5);
\draw[<-] (m5) -- (m6);
\draw[<-] (m6) -- (m7);
\draw[<-] (m7) -- (m1);

\draw[->] (m2) .. controls +(-1.2,0.3) .. (m1);
\draw[<-] (m3) -- (m1);
\draw[<-] (m4) -- (m1);
\draw[<-] (m5) -- (m1);
\draw[<-] (m6) -- (m1);

\end{scope}

\draw[->] (m1) edge [in=30,out=60,min distance=8] node {} (m1);
\draw[->] (m2) edge [in=75,out=105,min distance=8] node {} (m2);
\draw[->] (m3) edge [in=75,out=105,min distance=8] node {} (m3);
\draw[->] (m4) edge [in=85,out=115,min distance=8] node {} (m4);
\draw[->] (m7) edge [in=75,out=105,min distance=8] node {} (m7);

\draw[->] (n1) edge [in=120,out=150,min distance=8] node {} (n1);
\draw[->] (n2) edge [in=75,out=105,min distance=8] node {} (n2);
\draw[->] (n3) edge [in=75,out=105,min distance=8] node {} (n3);
\draw[->] (n4) edge [in=65,out=95,min distance=8] node {} (n4);
\draw[->] (n7) edge [in=75,out=105,min distance=8] node {} (n7);

\draw[->] (n7) .. controls +(-1.3, 0.4) .. (m7);
\draw[<-] (n7) .. controls +(-1.3,-0.4) .. (m7);
\end{tikzpicture}
\caption{Example of a network with exponential decision time of the equal
neighbor algorithm}
\label{fig:ex:exp}
\end{figure}
	 
\begin{thm}[\cite{OT11v1}]\label{thm:butterfly}
In the coordinated model consisting of  the $m$-butterfly graph, for any $\varepsilon >0$,
	the equal neighbor averaging algorithm does not achieve $\varepsilon$-consensus
	by round $K$ if $K = O\big( 4^{m/3}\log\frac{1}{\varepsilon} \big)$.
\end{thm}

The complete proof we give here is based on spectral gap arguments, and therein differes from the
        proof in~\cite{OT11v1}.
Let $W$ be any stochastic matrix, and let $\pi$ denote its Perron vector.
For any subset $S \subseteq [n]$, we let 
	\begin{equation*}\label{eq:def:bottleneck:set}
	\Phi_S(W) = \pi(S)^{-1} \sum_{p\in S} \sum_{q\notin S} \pi_q W_{pq}
	\end{equation*}
	where $\pi(S) = \sum_{r\in S} \pi_r$.
Then the {\em Cheeger constant\/} of the matrix~$W$ is defined as the minimal
	$\Phi_S(W)$ of all nonempty~$S$ with $\pi(S)\leq 1/2$, i.e.,
	\begin{equation*}\label{eq:def:bottleneck}
	\Phi(W) = \min_{\substack{S\subseteq[n]\\0<\pi(S)\leq1/2}} \Phi_S(W) \, .
	\end{equation*}
Combining Theorem 7.3 and Theorem 12.3 in~\cite{LPW09}, we obtain 
	the following Cheeger inequality: if~$\lambda$ denotes an
	eigenvalue of~$W$  other than 1 with the  greatest absolute value and $\pi_{\min}= \min_{p} ( \pi_p)$, then 
	\begin{equation}\label{eq:bottleneck}
	\lvert \lambda \rvert \geq
	1 - C\cdot \Phi(W) \cdot\log \frac{1}{\pi_{\min}}
	\end{equation}
	where $C$ is some universal constant.
	
\vspace{0.2cm}
We are now in position to prove Theorem~\ref{thm:butterfly}.

\begin{proof}
Let~$W$ be the stochastic matrix of size $n=2m$ associated with the equal-neighbor averaging 
	algorithm when the communication graph is the $m$-butterfly graph.
We verify that~$W$ is a primitive matrix and its Perron vector is given by
\begin{equation*}\label{eq:ex:chung:pi}
\pi_1 = \frac{1}{5}
\quad,
\quad
\pi_p = 
\frac{3}{5\cdot 2^{p}} \text{ for } p\in\{2,\dots,m-1\}
\quad\text{and}\quad
\pi_m = \frac{3}{5\cdot 2^{m-1}} \, .
\end{equation*}
By symmetry, this also defines the Perron vector for the
remaining indices between~$m+1$ and~$2m$ since $\pi_p =
\pi_{2m - p +1}$. 
Hence the smallest entry of $W$'s Perron vector is 
	$$\pi_{\min} = \pi_m = \frac{3}{5 \cdot 2^{m-1}} \, .$$
Choosing $S = \{1,2,\dots,m\}$, we have  $\pi (S) = 1/2 $, and
	$$\Phi_S(W) = 2 \pi_m W_{m\bar{m}} = \frac{1}{5\cdot 2^{m-2}} \, .$$
Hence 
	$\Phi(W) \leq \frac{1}{5\cdot 2^{m-2}} $.
The above Cheeger inequality~(\ref{eq:bottleneck}) gives the existence of an eigenvalue
	$\lambda\neq1$ of~$W$ with 
\begin{equation}\label{eq:ex:chung:gap}
1 - \lvert \lambda\rvert
=
O\left( \frac{m}{2^m} \right)
\end{equation}

We extend the definition of the vector semi-norm~$\delta$ to complex
	vectors by setting
	$$\delta(z) = 2\inf_{c\in\IC} \lVert z - c\cdot \onevec \rVert_\infty \, .$$ 
For all real vectors~$x$ and~$y$, we have
\begin{equation}\label{eq:delta:complex:real}
\max\{ \delta(x) , \delta(y) \}
\leq
\delta(x + iy)
\leq
2\max\{ \delta(x) , \delta(y) \} \, .
\end{equation}
Let $w$ be a $\delta$-normalized (possibly complex)
	eigenvector associated to~$\lambda$; thus
	$\delta(W^k w) = \lvert\lambda\rvert^k$.
Writing $w=u + iv$, we deduce from~\eqref{eq:delta:complex:real} that
	\begin{equation*}
	2\max\{\delta(W^k u),\delta(W^k v)\} \geq \delta(W^k w) = \lvert \lambda\rvert^k \delta(w)
	\geq  \lvert \lambda \rvert^k \max \{\delta(u),\delta(v)\} \, ,
	\end{equation*}
	which shows that either $\delta(W^k u)/\delta(u) \geq \lvert \lambda \rvert^k/2$
	or $\delta(W^k v)/\delta(v) \geq \lvert \lambda \rvert^k/2$.
This means $\delta(W^k) \geq \lvert \lambda\rvert^k/2$.
We use the inequality $1-a \geq e^{-2{a}}$, which holds for all $0\leq a \leq 0.7968$.
Together with~\eqref{eq:ex:chung:gap}, it shows
	\begin{equation*}
	\delta(W^k) = \exp\left(-O(k{m}/2^{m})\right) =  \exp\left(-O(k/2^{n/3})\right)
	\, .
	\end{equation*}
This means that for every~$\varepsilon>0$ and every $K = O\big( 2^{n/3} \log\frac{1}{\varepsilon} \big)$, 
	there exist initial values between~$0$ and~$1$ such that $\delta\big( x(K) \big) > \varepsilon$, i.e., 
	in the equal neighbor algorithm, processes should not decide at time~$K$ in order not to violate 
	$\varepsilon$-Agreement when the communication graph is the time-constant $m$-butterfly graph.
\end{proof}

\section{Approximate consensus with dynamic faults}

Time varying communication graph may result from benign communication faults in the
	case of message losses.
With such an interpretation of missing links in  communication graphs, Theorem~\ref{thm:SW} 
	coincides with the impossibility result of consensus established by Santoro and Widmayer~\cite{SW89}
	for synchronous systems with $n$ processes connected by  a complete communication graph and 
	$n-1$ communication faults per round.

In the light of Theorem~\ref{thm:coord}, we now revisit the problem of approximate consensus in the
	context of communication faults, whether they are due to  link or process failures.
We begin with a corollary of Theorem~\ref{thm:coord} that gives a new result  on the solvability 
	of approximate consensus in a complete synchronous network.
Then we explain how Theorem~\ref{thm:coord} provides a new understanding of some classical procedures 
	in approximate consensus algorithms to tolerate  benign process failures, and how it extends
	the  correctness of these procedures to {\em dynamic}  failure models.
We also derive the known result that approximate consensus is solvable in an asynchronous system
	with a complete communication graph and process failures if and only if there is a strict majority of non faulty
	processes.

\subsection{Link faults}

We begin with a simple sufficient condition for a directed graph to be rooted.
 
\begin{lem}\label{lem:2n-2}
Any directed graph with~$n$ nodes and  at
	least $n^2 - 3n + 3$ links is rooted.
\end{lem}
\begin{proof}
Let $G$ be a directed graph with~$n$ nodes that  is not a rooted graph.
Then the condensation of $G$ has   two nodes  without incoming link.
We denote the corresponding two strongly connected components in~$G$ by~$S_1$ and $S_2$,
	and their cardinalities by~$n_1$ and~$n_2$, respectively.
Therefore the number of links in~$G$ that are not self-loops is at most equal to $n^2-n - n_1 (n-n_1) - n_2 (n-n_2)$.
Since  for every pair of  integers $n_1,n_2$ in $[n-1]$, we have
	$$ n^2-n - n_1 (n-n_1) - n_2 (n-n_2) \leq  n^2 -3n +2 \, ,$$
	 it follows that $G$ has at most  $n^2 - 3n + 2$ links.
\end{proof}

From the equality $n^2 - 3n + 3 = (n^2 - n) - (2n -  3)$, we immediately derive the following theorem.

\begin{thm}\label{thm:2n-2}
Approximate consensus is solvable in a complete network with~$n$ processes if
	there are at most $2n-3$ link faults per  round.
\end{thm}

Interestingly, compared with the impossibility result  for consensus  in Theorem~\ref{thm:SW},  
	the above theorem shows that the number of link faults that can be tolerated increases by a factor of at least 2 
	when solving  approximate consensus instead of consensus.
Besides  it is easy to construct a non-rooted communication graph with  $n^2- 2n  + 2$ links which,
	combined with Theorem~\ref{thm:noncoord},  shows that the bound in the above theorem is tight.

\subsection{Dynamic sender faulty omission model}

Of particular interest are the failure models for complete networks in which process senders are blamed for message
	losses: in this way, one limits  the number of nodes which are origins of
	missing links in  communication graphs to some integer  $f \in [ n-1 ] $.
Parameter $f$ is not a global upper bound over the whole executions, but bounds the number of faulty senders 
	{\em round by round}.
The resulting failure model, referred to as the {\em dynamic sender faulty model}, thus handles dynamic
	failures. 
	
Theorem~\ref{thm:nonsplit} directly applies since the corresponding communication graphs are nonsplit, 
	even in the case $f=n-1$.
Moreover  following the proof of Theorem~\ref{thm:nonsplit}, at least $n-f$ columns of the stochastic matrix $W(k)$
	associated to round~$k$ of an equal neighbor averaging algorithm are positive; hence $W(k)$ satisfies
	$$ \delta(W(k)) \leq  \frac{f}{n} \, ,$$
	which shows that decisions can be made at round $\log_2  \frac{1}{\varepsilon} $ when only a minority
	of processes may be faulty.
Using the inequalities
	$$ \log \frac{n}{f}  \leq \log \left ( 1 + \frac{1}{n-1} \right ) \leq \frac{1}{n} \, ,$$
	we derive a linear time complexity in the number of processes in the wait-free case, i.e., $f=n-1$.

\begin{cor}\label{cor:clean}
In a complete synchronous network with $n$ processes and the dynamic sender faulty  model, 
	the equal neighbor averaging algorithm achieves $\varepsilon$-consensus in
	$ n \log \frac{1}{\varepsilon} $ rounds.
If in each round, only a minority of processes may be faulty, then processes can decide from  round
	$\log_2  \frac{1}{\varepsilon} $ onwards. 
\end{cor}
	
In the context of omission models with static faulty senders, different algorithmic procedures have been introduced 
	for solving approximate consensus with averaging algorithms (for instance, see ~\cite{DLPSW86,Fek90}).
We briefly recall two of them, namely the {\em  Reduce} and the {\em Center} procedures, and explain 
	how both correctness  and time complexity of the resulting
	approximate consensus algorithms immediately follow from Theorem~\ref{thm:nonsplit}
	 in the sender faulty omission model.

The first procedure is designed for the case $n>2f$.
For each process, it consists in replacing each missing value by 
	some arbitrary value smaller than 0, and then in computing the mean of the values 
	that remain  after removing the $f$ smallest values.
This procedure, called {\em Reduce}, corresponds to a {\em logical} communication graph
	that is also rooted: among the $n-f$ values selected by a process,
	at least $n-2f$ are actually selected by all the other processes.
Theorem~\ref{thm:nonsplit} ensures that the  averaging algorithms resulting from the Reduce procedure
	achieve approximate consensus.
For an equal neighbor averaging algorithm, the stochastic matrix $W(k)$ associated to round~$k$ 
	has at least $n-2f$ positive columns and satisfies 
	$$  \delta(W(k)) \leq  1-\frac{n-2f}{n-f} = \frac{f}{n-f}  \, .$$ 
From  $ f/n  \leq  f/( n-f ) $, we conclude that the Reduce procedure  slows down 
	equal neighbor averaging algorithms, even though time complexity remains linear:
	the Reduce procedure is useless in the context of omissions and has been actually  introduced to 
	tolerate Byzantine failures.
	
For tolerating $f$ crash failures, Fekete~\cite{Fek90} introduced another procedure, called {\em Center}, which 
	is actually a refinement of the Reduce procedure:  at every round, each process selects
	 the $n-f$ or $n-f+1$ central values it has just received.
More precisely, if process~$p$ receives $n-t$ values at round~$k$, then either $f-t$ is even and $p$ 
	removes  the $(f-t)/2$ smallest values and the $(f-t)/2$ greatest values, or 
	$f-t$ is odd and $p$  only removes the $(f-t-1)/2$ smallest values and the $(f-t -1)/2$ greatest values.
In the case $f-t$ is even, $p$ applies the equal neighbor averaging rule to update its local variable:
	all the selected values have the same weight, namely   $1/(n-f)$.
Otherwise $f-t$ is odd and $p$ computes the weighted average of the selected values with
	the same weight  $1/(n-f)$ for all values except the 
	smallest and the greatest one whose weight is half, namely $1/2(n-f)$.
As a corollary of Theorem~\ref{thm:coord}, we can prove the correctness of the Center procedure to solve
	approximate consensus in the dynamic sender faulty model.
Moreover we can check that the  stochastic matrix $W(k)$ associated to round~$k$  satisfies 
	$$  \delta \left( W(k) \right) \leq  \frac{f}{2(n-f)}  \, .$$
Therefore the Center procedure improves the simple equal neighbor averaging algorithms if $2f <n$.  
The results in~\cite{Fek90} concerning crash failures can thus be directly derived from
	Theorem~\ref{thm:coord} and  from the translation of the Center procedure in terms of averaging 
	algorithms, and thus are noticeably extended to the dynamic sender faulty model.

\subsection{Asynchronous systems with crash failures}

We now consider asynchronous complete networks with $n$ processes among which at most
	$f$ may crash.
As observed in~\cite{CS09}, in such networks we can easily  implement 
	 communication graphs~$G$ such that for each process~$p$,
	$$ | \In_p (G) | = n-f \, . $$
If only a minority of processes may crash, i.e., $n>2f$, we thus obtain a nonsplit network model,
	and Theorem~\ref{thm:nonsplit} applies.	
In particular, approximate consensus can be solved with an equal neighbor averaging algorithm that 
	terminates in a linear number of  rounds.	
Observe that   time complexity drastically reduces in synchronous networks to the constant $\log_2  \frac{1}{\varepsilon} $
	in Corollary~\ref{cor:clean}.

\begin{cor}\label{cor:asynch}
In an asynchronous complete network of $n$ processes among which a minority may crash, 
	nonsplit rounds can be implemented and the equal neighbor averaging 
	algorithm 	in which all non-crashed processes decide at  round
	$n  \log  \frac{1}{\varepsilon} $ 
	achieves $\varepsilon$-consensus.
\end{cor}

The equal neighbor averaging algorithm  coincides with the {\em AsynchApproxAgreement}
	algorithm proposed in~\cite{Lyn96} to solve the approximate consensus problem in the
	case $n>3f$: Lynch claimed that ``a more complicated algorithm is needed for $n>2f$.''
Corollary~\ref{cor:asynch} shows that this algorithm actually works while  $n>2f$. 

Finally a simple partitioning argument shows that approximate consensus is not 
	solvable if $n\leq 2f$.

%
%
%
%
%

\section{Averaging algorithms in partially synchronous systems}

The round-based computational model considered so far assumes that rounds are communication 
	closed layers: messages from one process to another are delivered in the rounds in which 
	they are sent.
To guarantee the latter condition, processes just need to timestamp the messages 
	they send with  the current round number, and to discard old messages, 
	i.e., messages sent in previous rounds.
In  a perfect synchronous network with  transmission delays upper bounded by~$D$, rounds are implemented in
	an optimal way when using  timeouts that are
	equal to $D$;  in this way, no messages are discarded.
If the network is not perfectly synchronous, an aggressive politics for timeouts may result in
	discarding many messages while large timeouts drastically slow down the system.
To manage the trade-off between timeliness and connectivity, one may relax the condition
	of communication closed layers by allowing processes to receive {\em outdated} messages.
However the number of rounds between the sending and the receipt of messages 
	should be bounded  to keep the system efficient.
This notion of  {\em partially synchronous rounds} coincides with the model of distributed asynchronous computation
	developed by Tsitsiklis in~\cite{Tsi84,BT89}.

In the case of averaging algorithms with a bound on message delays equal to $\Delta$, the local variable~$x_p$ is updated according to the
	following rule:
	\begin{align}
 	x_p(k) = \sum_{q \in \In_p(k)} w_{qp}(k) \, x_q \left (  \kappa_q^p(k) \right ) ,\label{eq:updatedelay}
	\end{align}
	where $ \kappa_q^p(k) \in \{ k - \Delta, \dots, k-1\}$.
Since each process~$p$ has  immediate access to its own local variable $x_p$,  we further
	assume that for every partially synchronous round~$k$,
	$$   \kappa_p^p(k) =k - 1\, .$$
We call such an execution a {\em $\Delta$-bounded execution}.
The case of zero communication delays is captured by $\Delta = 1$, and equation (\ref{eq:updatedelay}) corresponds 
	in this case to an execution of the averaging algorithm with synchronous rounds.
We do not require the functions $\kappa_q^p$ to be either non-decreasing, surjective, or injective.
In other words, communications between processes may be non-FIFO and unreliable (duplication and loss). 

Note that the communication graph~$G(k)$ in round~$k$ is understood to be the
	graph defined by the incoming values at round~$k$, i.e., $(p,q)$ is a link in $G(k)$ 
	if and only if $w_{qp}(k)>0$.

We now extend Theorems~\ref{thm:nonsplit} and~\ref{thm:coord} to
	partially synchronous rounds.
Our proof strategy is based on a reduction to the synchronous case:
	each process corresponds to a set of $\Delta$ virtual processes, and 
	every  $\Delta$-bounded execution of an averaging algorithm 
	with $n$ processes coincides with a synchronous execution
	of an averaging algorithm with~$n\Delta$ processes.

\subsection{Reduction to synchronous rounds}\label{sec:red}

We mimic the  reduction of a $\Delta$-th order ordinary differential
	equation to a system of $\Delta$ ordinary differential equations of first order.	
We define the vectors $\tilde{x}(k)\in\IR^{n\Delta}$ by setting
	\begin{equation}
	\tilde{x}_{p\Delta-d}(k) = x_p(k-d)
	\end{equation}
	for $p\in[n]$,  $0\leq d \leq \Delta-1$, and with the auxiliary definition $x(-k) = x(0)$ for all
	the positive integers $k$.
We also define the  the $\Delta n \times \Delta n$ matrix  
	$\tilde{W}(k) $ by
	\begin{equation}\label{eq:tildeW}
	\tilde{W}_{p\Delta-d,q\Delta-d'} (k)
	=
		\begin{cases}
		w_{qp}(k) & \text{if $d=0$ and $d' = k - \kappa_{q}^{p}(k) -1 $}\\
		1 & \text{if $p=q$ and $d' = d - 1$}\\
		0 & \text{else}
		\end{cases}
	\end{equation}	
The key point is  that  the vector $ \tilde{x}(k)$ is updated according to the linear
	recursion with ``zero delay'' 
	\begin{equation}\label{eq:recursion:partial}
	\tilde{x}(k) = \tilde{W}(k) \tilde{x}(k-1)\, ,
	\end{equation}
	 which we prove now.
For $d\neq 0$, we have
\begin{equation*}
\begin{split}
\big(\tilde{W}(k)\tilde{x}(k-1)\big)_{p\Delta - d} & =
\sum_ b \tilde{W}_{p\Delta - d, b} (k) \cdot \tilde{x}_{ b}(k-1)
=
\tilde{x}_{p\Delta - (d -1)}(k-1)
\\ & =
x_p\big(k - 1 - (d -1)\big) = x_p(k-d)
=
\tilde{x}_{p\Delta - d}(k) \, ,
\end{split}
\end{equation*}
while
for $d = 0$, we have
\begin{equation*}
\begin{split}
\big(\tilde{W}(k)\tilde{x}(k-1)\big)_{p\Delta}
& =
\sum_ b \tilde{W}_{p\Delta, b} (k) \cdot \tilde{x}_{ b}(k-1)
=
\sum_q w_{qp}(k) \cdot \tilde{x}_{q\Delta - (k - \kappa_q^p(k) - 1)}(k-1)
\\ & =
\sum_q w_{qp}(k) \cdot x_q\big( \kappa_q^p(k)\big) 
=
x_p(k) 
=
\tilde{x}_{p\Delta}(k) \, .
\end{split}
\end{equation*}
In any case, we have $\big(\tilde{W}(k)\tilde{x}(k-1)\big)_{ b} = \tilde{x}_ b(k)$, 
	which shows~\eqref{eq:recursion:partial}.

We easily check that each matrix $\tilde{W}(k) $ is stochastic  with positive entries 
	 at most equal to~$\varrho$.
Because of (\ref{eq:tildeW}), we have 
	\begin{equation}\label{eq:tildeloop}
	 \tilde{W}_{p\Delta,p\Delta} (k) = w_{pp}(k) \, ,
	 \end{equation}
	which is positive since there is a self-loop at~$p$ in the communication graph~$G(k)$.

\begin{figure}
\centering
\begin{tikzpicture}[>=latex',scale=0.8, every node/.style={transform shape}]
\node[draw,circle] (n0) at (0,0)  {$\scriptstyle 5p-0$};
\node[draw,circle] (n1) at (-2,0) {$\scriptstyle 5p-1$};
\node[draw,circle] (n2) at (-4,0) {$\scriptstyle 5p-2$};
\node[draw,circle] (n3) at (-6,0) {$\scriptstyle 5p-3$};
\node[draw,circle] (n4) at (-8,0) {$\scriptstyle 5p-4$};
\draw[<-] (n0) .. controls +(1.0,-0.5) and +(1.0,0.5) .. node[right]
{$w_{pp}(k)$} (n0);
\draw[<-] (n1) -- node[below] {$1$} (n0);
\draw[<-] (n2) -- node[below] {$1$} (n1);
\draw[<-] (n3) -- node[below] {$1$} (n2);
\draw[<-] (n4) -- node[below] {$1$} (n3);
\node[draw,circle] (m0) at ( 0,-2.5) {$\scriptstyle 5q-0$};
\node[draw,circle] (m1) at (-2,-2.5) {$\scriptstyle 5q-1$};
\node[draw,circle] (m2) at (-4,-2.5) {$\scriptstyle 5q-2$};
\node[draw,circle] (m3) at (-6,-2.5) {$\scriptstyle 5q-3$};
\node[draw,circle] (m4) at (-8,-2.5) {$\scriptstyle 5q-4$};
\draw[<-] (m0) .. controls +(1.0,-0.5) and +(1.0,0.5) .. node[right]
{$w_{qq}(k)$} (m0);
\draw[<-] (m1) -- node[below] {$1$} (m0);
\draw[<-] (m2) -- node[below] {$1$} (m1);
\draw[<-] (m3) -- node[below] {$1$} (m2);
\draw[<-] (m4) -- node[below] {$1$} (m3);
\draw[->] (n2) -- node[right] {$w_{qp}(k)$} (m0);
\end{tikzpicture}
\caption{Part of the graph $\tilde{G}(k)$ corresponding to link $(p,q)$ in $G(k)$, 
	$\Delta=5$, and $\kappa_q^p(k) = 2$}
\label{fig:delay:reduction}
\end{figure}
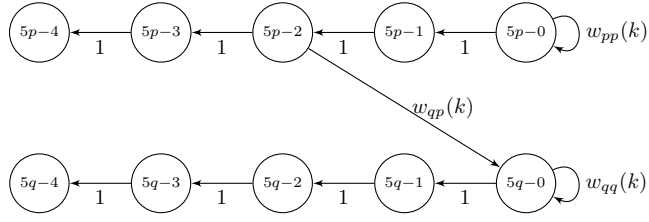

Let $\tilde{G}(k) = \big([n\Delta] , \tilde{E}(k)\big)$ be the directed graph associated
	to the stochastic matrix $ \tilde{W}(k)$, i.e.,  there is a link
	$( a, b)$ in  $\tilde{G}(k) $ if and only if $\tilde{W}_{ b, a}(k)>0$.
The graph $\tilde{G}(k) $ can be seen as the communication graph of a network of
	$n\Delta$ processes with the restriction that some nodes have no self-loop.
By~(\ref{eq:tildeloop}), there is however a self-loop at each  node $p\Delta$ with $p\in[n]$.
Figure~\ref{fig:delay:reduction} shows part of  $\tilde{G}(k)$  corresponding to
	one  link $(p,q)$ in $G(k)$.


\subsection{Nonsplit network model}

Even if the communication graph $G(k)$ is nonsplit, 
	 the graph~$\tilde{G}(k)$ contains two nodes without a common 
	incoming neighbor when $\Delta >1$.
However we will show that in a nonsplit network model, each cumulative graph over $2\Delta -1$ 
	rounds is nonsplit, which allows us to extend Theorem~\ref{thm:nonsplit} to the
	partially synchronous case.

\begin{thm}\label{thm:nonsplit:delay}
In a nonsplit network model of $n$ processes, every averaging algorithm with 
	parameter~$\varrho$ achieves $\varepsilon$-consensus
	in $(2\Delta-1)\big(\frac{1}{\varrho}\big)^{2\Delta-1}\log \frac{1}{\varepsilon} + 2\Delta-2$
	rounds of a $\Delta$-bounded execution.
\end{thm}

\begin{proof}
Validity is clear because of the fact that every value in~$\tilde{x}(k)$, and
	hence in~$x(k)$ is a convex combination of values in~$x(0)$.

For $\varepsilon$-Agreement,  we first show that each matrix $\tilde{W}(k +2\Delta-2:k)$ 
	is nonsplit.
We consider its associated graph $$\tilde{G}(k +2\Delta-2:k) = \tilde{ G }(k)\circ \dots \circ \tilde{ G}(k +2\Delta-2)$$ 
	and two arbitrary nodes $a = p\Delta -d$  and $b = q\Delta -d'$.
Since  the original communication graph~$G(k \Delta-1)$ is nonsplit,
	$p$ and $q$ have a common incoming neighbor~$r$ in this graph, i.e.,
	for some $r_1 =  r \Delta -d_1$ and $r_2 =  r \Delta -d_2$, there is a link from $r_1$ to $p\Delta$
	and a link  from $r_2$ to $q\Delta$ in the directed graph~$ \tilde{ G}(k +\Delta-1)$.
Following the links of the graphs $ \tilde{ G }(k),  \dots ,  \tilde{ G}(k +2\Delta-2)$ 
	depicted in~Figure~\ref{fig:ofMatthias},  we obtain a path from $r \Delta $ to $a$ 
	with $\Delta - d_1 -1$ self-loops at node~$r \Delta$, and $\Delta - d -1$ self-loops at 
	node~$p \Delta$.
This directed path corresponds to a link from $r \Delta $ to $a$ in~$\tilde{G}(k +2\Delta-2:k)$.
In the same way, we have a link from $r \Delta $ to $b$ in~$\tilde{G}(k +2\Delta-2:k)$, 
	which proves that the latter graph is nonsplit. 

\begin{figure}
\centering
\begin{tikzpicture}[>=latex',scale=0.8, every node/.style={transform shape}]
\node[draw,circle] (l0) at ( 0,2.5) {$\scriptstyle 5p-0$};
\node[draw,circle] (l1) at (-2,2.5) {$\scriptstyle 5p-1$};
\node[draw,circle] (l2) at (-4,2.5) {$\scriptstyle 5p-2$};
\node[draw,circle] (l3) at (-6,2.5) {$\scriptstyle 5p-3$};
\node[draw,circle] (l4) at (-8,2.5) {$\scriptstyle 5p-4$};
\draw[<-] (l0) .. controls +(1.0,-0.5) and +(1.0,0.5) .. node[above
right=3pt,shape=circle,draw,inner sep=2pt]
{$5$} (l0);
\draw[<-] (l1) -- node[above,shape=circle,draw,inner sep=2pt] {$6$} (l0);
\draw[<-] (l2) -- node[above,shape=circle,draw,inner sep=2pt] {$7$} (l1);
\draw[<-] (l3) -- node[above,shape=circle,draw,inner sep=2pt] {$8$} (l2);
\draw[<-] (l4) -- node[below] {} (l3);

\node[draw,circle] (n0) at (0,0)  {$\scriptstyle 5r-0$};
\node[draw,circle] (n1) at (-2,0) {$\scriptstyle 5r-1$};
\node[draw,circle] (n2) at (-4,0) {$\scriptstyle 5r-2$};
\node[draw,circle] (n3) at (-6,0) {$\scriptstyle 5r-3$};
\node[draw,circle] (n4) at (-8,0) {$\scriptstyle 5r-4$};
\draw[<-] (n0) .. controls +(1.0,-0.5) and +(1.0,0.5) ..
node[above right=3pt,shape=circle,draw,inner sep=2pt] {$0$}  (n0);
\draw[<-] (n1) -- node[above,shape=circle,draw,inner sep=2pt] {$1$} 
node[below,shape=rectangle,draw,inner sep=2pt] {$0$} (n0);
\draw[<-] (n2) -- node[above,shape=circle,draw,inner sep=2pt] {$2$}
node[below,shape=rectangle,draw,inner sep=2pt] {$1$} (n1);
\draw[<-] (n3) -- node[above,shape=circle,draw,inner sep=2pt] {$3$}
node[below,shape=rectangle,draw,inner sep=2pt] {$2$} (n2);
\draw[<-] (n4) -- node[below] {} node[below,shape=rectangle,draw,inner sep=2pt]
{$3$} (n3);

\node[draw,circle] (m0) at ( 0,-2.5) {$\scriptstyle 5q-0$};
\node[draw,circle] (m1) at (-2,-2.5) {$\scriptstyle 5q-1$};
\node[draw,circle] (m2) at (-4,-2.5) {$\scriptstyle 5q-2$};
\node[draw,circle] (m3) at (-6,-2.5) {$\scriptstyle 5q-3$};
\node[draw,circle] (m4) at (-8,-2.5) {$\scriptstyle 5q-4$};
\draw[<-] (m0) .. controls +(1.0,-0.5) and +(1.0,0.5) .. node[right]
{} (m0);
\draw[<-] (m1) -- node[below,shape=rectangle,draw,inner sep=2pt] {$5$} (m0);
\draw[<-] (m2) -- node[below,shape=rectangle,draw,inner sep=2pt] {$6$} (m1);
\draw[<-] (m3) -- node[below,shape=rectangle,draw,inner sep=2pt] {$7$} (m2);
\draw[<-] (m4) -- node[below,shape=rectangle,draw,inner sep=2pt] {$8$} (m3);

\draw[->] (n3) -- node[above,shape=circle,draw,inner sep=2pt] {$4$} (l0);
\draw[->] (n4) -- node[below,shape=rectangle,draw,inner sep=2pt] {$4$} (m0);
\end{tikzpicture}
\caption{Paths from $r\Delta$ to $a = p\Delta-3$ and $b = q\Delta - 4$ in
	the cumulative  graph $\tilde{G}(k+2\Delta -2 : k)$  with $\Delta=5$.
The link labels~$\ell$ (circled for the path to~$a$, boxed for
the path to~$b$)  denote the use of that link in round~$k+\ell$.}
\label{fig:ofMatthias}
\end{figure}
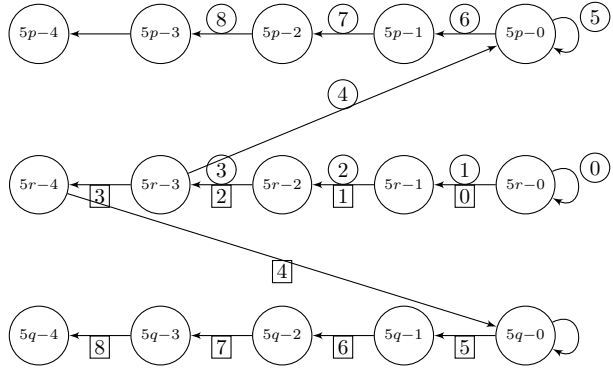
	
Since each positive entry of 
	$\tilde{W}(k+2\Delta-2:k)$ is at most equal to  $\varrho^{2 \Delta -1}$,
	 the recurrence relation~(\ref{eq:updatedelay}) thus corresponds to
	the synchronous execution of an averaging  algorithm with
	parameter $\varrho^{2 \Delta -1}$ in a generalized
	$2 \Delta -1$-nonsplit network model of $n\Delta$ processes.

Since $\delta\big( \tilde{x}(0) \big) = \delta\big( x(0) \big) \leq 1$, we deduce from Corollary~\ref{cor:nonsplit} that 
	 if   
	$k \geq (2\Delta \! - \!1) \left( \frac{1}{\varrho} \right)^{2\Delta \!- \!1} \!\! \log \frac{1}{\varepsilon} + 2\Delta -~2$,
	then     $\delta\big( \tilde{x}(k) \big) \leq \varepsilon $.
Observing $\delta\big( x(k) \big) \leq \delta\big( \tilde{x}(k)\big)$ for each positive integer~$k$ 
	shows $\varepsilon$-Agreement, which completes the proof.	
\end{proof}

\subsection{Coordinated network model}

Contrary to the nonsplit property,~$\tilde{G}(k)$ is rooted whenever~$G(k)$ is.
However some nodes in~$\tilde{G}(k)$ have no self-loops, and so we cannot apply
	Theorem~\ref{thm:coord} to the averaging algorithm and the virtual network model   
	corresponding to~(\ref{eq:recursion:partial}).
Then we will use another strategy which consists in proving that  each cumulative graph over $n\Delta$ 
	rounds is nonsplit. 
	
\begin{thm}\label{thm:coord:delay}
In a coordinated network model, every averaging algorithm with parameter~$\varrho$ achieves $\varepsilon$-consensus
in $n\Delta \left(\frac{1}{\varrho} \right)^{n\Delta} \log
\frac{1}{\varepsilon} + n\Delta-1$ rounds of a $\Delta$-bounded execution.
\end{thm}
\begin{proof}
Validity is clear because of the fact that every value in~$\tilde{x}(k)$, and
	hence in~$x(k)$ is a convex combination of values in~$x(0)$.

For $\varepsilon$-Agreement,  we first show that each matrix $\tilde{W}(k
+n\Delta-1:k)$ 
	is nonsplit.
We consider its associated graph $$\tilde{G}(k +n\Delta-1:k) = \tilde{ G }(k)\circ \dots \circ \tilde{ G}(k +n\Delta-1)$$ 
	and two arbitrary nodes $a = p\Delta -d$  and $b = q\Delta -d'$.
This then proves the theorem by Corollary~\ref{cor:nonsplit}.

\medskip

Let $k$ and $\ell$ be two positive integers such that $\ell \geq K$.
For every $b\in[n\Delta]$, we define the sets
\begin{equation*}
 \tilde{S}_{b}(\ell)
=
\{ a\in[\Delta n] \mid (a,b) \text{ is a link of } \tilde{G}(\ell:k) \}
\end{equation*}
and for every $c\in[n\Delta]$,
\begin{equation*}
T_c(\ell) 
=
\{ p\in[n] \mid (p\Delta,c) \text{ is a link of } \tilde{G}(\ell:k) \} \, .
\end{equation*}
Because each $(p\Delta-\delta,p\Delta-\delta+1)$ with $0\leq \delta\leq \Delta - 1$ 
	is a link of  all graphs~$\tilde{G}(k)$, we have $p\in T_a(k +d)$
and
$q\in T_b(k+d')$.
The existence of the self-loops at the nodes~$r\Delta$ for all $r\in[n]$
implies the monotonicity property $T_c(\ell) \subseteq T_c(\ell + 1)$, which in
particular gives
\begin{equation}\label{eq:T:init}
p\in T_a(k + \Delta - 1) 
\quad \text{and} \quad
q\in T_b(k + \Delta - 1) \, .
\end{equation}
Mimicking the proof of Theorem~\ref{thm:coord}, we show that 
\begin{equation}\label{eq:Ts:intersect}
T_{a}(k + n\Delta - 1)
\cap
T_{b}(k + n\Delta - 1)
\neq
\emptyset
\enspace,
\end{equation}
which then concludes the proof.

If $p=q$, then~\eqref{eq:Ts:intersect} clearly holds by~\eqref{eq:T:init}.
Otherwise, consider the nondecreasing sequences
	$$T_{a}(k+\Delta-1)\subseteq T_{a}(k+2\Delta-1)\subseteq\cdots\subseteq T_{a}(k+n\Delta-1)\, , $$ 
	$$	T_{b}(k+\Delta-1)\subseteq T_{b}(k+2\Delta-1)\subseteq\cdots\subseteq T_{b}(k+n\Delta-1) \, ,$$
and 
\begin{equation*}
\begin{split}
T_{a}(k+\Delta-1)
\cup
T_{b}(k+\Delta-1)
&
\subseteq
T_{a}(k+2\Delta-1)
\cup
T_{b}(k+2\Delta-1)
\\ &
 \subseteq \cdots \subseteq 
T_{a}(k + n\Delta - 1)
\cup
T_{b}(k + n\Delta - 1)
\, .
\end{split}
\end{equation*}

Because $\lvert T_{a}(k+\Delta-1)
\cup
T_{a}(k+\Delta-1)\rvert\geq 2$ if $p\neq q$ and
$\lvert T_{a}(k + n\Delta - 1)
\cup
T_{b}(k + n\Delta - 1)\rvert \leq n$, the latter sequence cannot be strictly
increasing.
There hence exists some $\ell\in\{1,\dots,n-1\}$ such that
\begin{equation*}
T_{a}\big(k + \ell \Delta - 1\big) 
\cup
T_{b}\big(k + \ell \Delta - 1\big) 
=
T_{a}\big(k + (\ell+1) \Delta - 1\big) 
\cup
T_{b}\big(k + (\ell+1) \Delta - 1\big) 
\enspace.
\end{equation*}
Because of the disjointness assumption, this means 
\begin{equation*}
T_{a}\big(k + \ell \Delta - 1\big) 
=
T_{a}\big(k + (\ell+1) \Delta - 1\big) 
\quad\text{and}\quad
T_{b}\big(k + \ell \Delta - 1\big) 
=
T_{b}\big(k + (\ell+1) \Delta - 1\big) 
\enspace.
\end{equation*}
We now show that both 
$T_{a}\big(k + \ell \Delta - 1\big)$ and $T_{b}\big(k + \ell
\Delta - 1\big)$ contain the roots of $G(k+\ell\Delta)$,
showing~\eqref{eq:Ts:intersect}.

Suppose not, i.e., without loss of generality $T_{a}\big(k + \ell
\Delta - 1\big)$ does not contain all roots of $G(k+\ell\Delta)$.
By Proposition~\ref{prop:rooted}, there is an edge $(p,q)$ in
$G(k+\ell\Delta)$ such that
\begin{equation}\label{eq:roots:delays}
q\in T_{a}\big(k + \ell
\Delta - 1\big)
\quad\text{and}\quad
p\not\in T_{a}\big(k + \ell
\Delta - 1\big)
\enspace.
\end{equation}
There hence exists some $\delta\in\{0,\dots,\Delta-1\}$ such that
$(p\Delta-\delta,q\Delta)$ is a link  of $\tilde{G}(k+\ell\Delta)$, which
means that
\begin{equation*}
p\Delta-\delta \in \tilde{S}_{a}(k+\ell\Delta)
\enspace.
\end{equation*}
It follows that $p\Delta\in \tilde{S}_{
a}(k+\ell\Delta+\Delta-1)$, which means
\begin{equation*}
p\in T_{a}(k + (\ell+1)\Delta-1) = T_{a}(k + \ell\Delta-1)
\enspace.
\end{equation*}
This is a contradiction to~\eqref{eq:roots:delays} and concludes the proof.
\end{proof}

Theorem~\ref{thm:coord:delay} shows that, even in the class of $\Delta$-bounded
	executions, approximate consensus is solvable in a network model~$\mathcal{N}$
	if~$\mathcal{N}$ is a coordinated model.
By Theorem~\ref{thm:noncoord}, the latter condition is also necessary for the subset of 
	1-bounded executions, and so {\it a fortiori} for $\Delta$-bounded executions.
The characterization of the network models in which approximate consensus is solvable 
	in Corollary~\ref{cor:characterization} then holds for computations with
	partially synchronous rounds as well as with  synchronous rounds.  

\section{Conclusion and Future Work}

The main goal of this paper has been to characterize the dynamic
     network models in which approximate consensus is solvable.
Interestingly anonymity of processes does not affect solvability in
     such networks.
We have further established some upper bounds on the time complexity
     of averaging algorithms, all of which solve approximate consensus
     in dynamic networks.
We have proved each of our computability and complexity results first
     for synchronous rounds and in a second step for partially
     synchronous rounds which allow for bounded message delays.

As for exact consensus, approximate consensus does not require strong
     connectivity and  it can be solved under the sole assumption of
     rooted communication graphs.
However contrary to the condition of a stable set of roots and
     identifiers supposed in~\cite{BRS12} for achieving consensus,
     approximate consensus can be solved even though roots arbitrarily
     change over time and processes are anonymous.
In these respects, approximate consensus seems to be more suitable than
     consensus for handling real network dynamicity.
	
A number of questions are suggested by this work.
For example, it would be of high interest to design approximate consensus algorithms
	that tolerate Byzantine process failure, i.e., arbitrary process behaviors.
Certain interesting questions also remain to be studied in the benign case.
In particular general lower bounds on the time complexity of approximate
consensus, be it for general algorithms or averaging algorithms would vastly
improve the comprehension of the approximate consensus problem.

\paragraph{Acknowledgments.} We wish to thank Alex Olshevsky for very helpful discussions on 
his work on consensus sets and Martin Perner and Martin Biely for many detailed comments.
\vspace{0.7cm}

\bibliographystyle{plain}

\bibliography{agents}

\end{document}